\documentclass[11pt,twoside]{article}

\usepackage{amsmath,amsfonts,amsthm,graphicx,amscd,amssymb,latexsym,mathrsfs,graphics,color,longtable,multicol,verbatim,lscape}
\usepackage[mathscr]{euscript}
\usepackage{enumerate}

\usepackage[all]{xy}

\usepackage[bookmarks,colorlinks,linkcolor=black,citecolor=black]{hyperref}

\numberwithin{equation}{section}

 \font\chuto=cmbx10 at 14pt

\font\ita=cmti9

\font\ita=cmti9   

\newtheorem{thm}{Theorem}[subsection]
\newtheorem{prop}[thm]{Proposition}

\theoremstyle{definition}
\newtheorem{defn}[thm]{Definition}
\newtheorem{rem}[thm]{Remark}
\newtheorem{assum}[thm]{Assumption}
\newtheorem{dataso}[thm]{Data sources}
\newcommand{\R}{\mathbb{R}}

\title{\vspace{-2cm}
	\vspace{2cm}{\chuto Some Applications of Lie Groups in Theory of Technical Progress}
	\footnotetext{
		{\bf MSC (2010):} \footnotesize{Primary 22E70;  Secondary 91B62, 91B38.\\
			{\bf Financial Support:} The paper was supported by the University Economics and Law, Viet Nam National University - Ho Chi Minh City, under grant C2019-34-07.
}}} 
\author{{\bf Le Anh Vu$^*$, Duong Quang Hoa$^\dagger$,}\\
	{\bf Nguyen Minh Tri$^{**}$, Ha Van Hieu$^{\ddag}$}\\
	\vspace{-0.2cm}{\ita $^*$ University of Economics and Law, VNU-HCMC, Viet Nam}\\
	\vspace{-0.2cm}{\ita Email: vula@uel.edu.vn}\\
	\vspace{-0.2cm}{\ita $^\dagger$ Hoa Sen University, Ho Chi Minh City, Viet Nam}\\
	\vspace{-0.2cm}{\ita Email: hoa.duongquang@hoasen.edu.vn}\\
	\vspace{-0.2cm}{\ita ${}^{**}$ University of Economics and Law, VNU-HCMC, Viet Nam}\\
	\vspace{-0.2cm}{\ita Email: trinm@uel.edu.vn}\\
	\vspace{-0.2cm}{\ita ${}^{\ddag}$ University of Economics and Law, VNU-HCMC, Viet Nam}\\
	\vspace{-0.2cm}{\ita Email: hieuhv@uel.edu.vn}}

\begin{document}

\date{}
\maketitle

\begin{abstract}
    In recent decades, we have known some interesting applications of Lie theory in the theory of technological progress. Firstly, we will discuss some results of R. Saito in \cite{rS1980} and \cite{rS1981} about the application modeling of Lie groups in the theory of technical progress. Next, we will describe the result on Romanian economy of G. Zaman and Z. Goschin in \cite{ZG2010}. Finally, by using Sato's results and applying the method of G. Zaman and Z. Goschin, we give an estimation of the GDP function of Viet Nam for the 1995-2018 period and give several important observations about the impact of technical progress on economic growth of Viet Nam.
\end{abstract}

\maketitle

\section*{Introduction and Motivation}
The study of the influence of technological progress on economic growth was first introduced in 1957 by Solow (see \cite{mS1957}) when he showed that a substantial portion of the increase in per capita output in the United States cannot be explained by growth in the capital-labor ratio. The unexplained portion is attributed to ``technical progress".
	
Clearly, the initial concept of ``technical progress" is quite broad and ambiguous, everything (except labor and capital) can be considered ``technical progress". In this conglomerate of factors that enter into ``technical progress'', there is one which is of special interest for the present discussion, and that is the effects of scale.
	
Problem distinguishing between ``scale effect'' and ``technical progress'' is the motivation of this topic. In 1980 and 1981, R. Saito \cite{rS1980}, \cite{rS1981} summarized and introduced some effective applications of Lie theory in solving this problem as well as in the theory of technical progress and economic growth, in general.
        
Basically, there are two directions to build models for studying the impact of technical progress on economic growth based on the views: technical progress is endogenous or exogenous. The endogenous or exogenous nature of the technical progress refers to its source: the endogenous change is internal to the national economy,
being created by domestic private or public enterprise, while the exogenous change is external, originating from foreign sources (see {\cite[Introduction]{ZG2010}}).

In 2010, G. Zaman and Z. Goschin \cite{ZG2010} developed several models to estimate the aggregated production function GDP of Romania for the 1990-2007 period, in both directions mentioned above to assess the impact of technical progress on the economic growth of Romania. 

In this paper, by using Sato's results and applying the method which is similar to the one of G. Zaman and Z. Goschin, we would like to study the impact of technical progress on economic growth of Viet Nam. Namely, we use data provided by the General Statistics Office of Viet Nam for the  1995--2018 period to analyze the impact of technical growth on the economy of Viet Nam. The main result of the paper is Theorem \ref{T233} in which we give an estimation of the GDP function of Viet Nam for the 1995-2018 period and several important observations about the impact of technical progress on the economic growth of Viet Nam.
		
The paper is organized as follows. In Section 1, we introduce some concepts and Sato's results in \cite{rS1980} and \cite{rS1981}. In the first part of Section 2, we will reiterate the result of G. Zaman and Z. Goschin as an illustrative example. Finally, we present and analyse the main result about the aggregated production function GDP of Viet Nam.		

\section{Preliminaries}\label{Sec1}     

In this section, we will recall related notions and discuss some results in \cite{rS1980} and \cite{rS1981}, before going into the main results in Section 2.

\subsection{The Product Function (\cite{rS1981}, Section 2)}
Consider a general (strictly) quasi-concave and a continuously differentiable neoclassical production function with the usual properties,
     \begin{equation}\label{EQ1}
                Y = f(K, L) 
     \end{equation}
where $ Y $ is the output, $ K $ is the capital, $ L $ is the labor of the production process. Of course, $K > 0$, $L > 0$.
            
Although the exogenous changes that affect the factor combination in a production process may result from many different forces such as new inventions and new applications of known technology, we may simply             identify them as ``technical change'' or ``technical progress'' and represent it by a parameter $t$ (time). Here, of course $t$ belongs in a finite subset of straight line $\R$.

Assume that when exogenous technical progress is introduced, it will not change the form of the production function, but it will change the output level by affecting the way in which the factor inputs are combined, i.e.,
     \begin{equation}\label{EQ2}
                Y_{t} = f(K, L, t) = f\left(K_{t}, L_{t}\right) \tag{{\ref{EQ1}}'}
     \end{equation}
    
\subsection{Technical Progress Functions (\cite[Section 2]{rS1981})} 
     \begin{defn} When exogenous technical progress $ t $ is introduced, it will change the way in which $ K $ and $ L $ are combined. The family $T : = \{T_t\}$ of pairs $T_t = (\bar{K}, \bar{L})$, where $\bar{K} = \varphi (K, L, t) = \varphi_{t}(K, L)$ and $\bar{L} = \psi (K, L, t)\bigr) = \psi_{t}(K, L)$ are the functions which combine the factor inputs through the technical progress parameter $ t $, is called {\it a family of technical progress functions} of $ K $ and $ L $ or simply {\it a technical progress}. 
     	
     In other words, a family $T = \{T_t\}$ of technical progress functions of $ K $ and $ L $ is given as follows 
            \begin{equation}\label{EQ:Tt}
                    \left\{T_{t} := \bigl( \varphi_{t}(K, L), \psi_{t}(K, L) \bigr) \right\}
                    \text{ or simply } T_{t} := \bigl( \varphi_{t}, \psi_{t} \bigr)
                \end{equation}
             where $\bar{K} = \varphi_{t}(K, L) = \varphi (K, L, t)$ is exactly the new capital (``effective" capital) and $\bar{L} =  \psi_{t}(K, L) = \psi (K, L, t)$ is exactly new labor (``effective" labor) when technical change has been integrated into them.
 
For any technical progress $T = \{T_t\}$, the components $\varphi = \varphi (K, L, t)$ and $\psi = \psi (K, L, t)\bigr)$ are always supposed to be generally analytic, real functions of the three variables $ K $, $ L $, and $ t $. Besides, $ \varphi $ and $ \psi $   are independent functions with respect to $ K $ and $ L $ alone, i.e.,
	\[ 
		\left|
			\def\arraystretch{2.2}
				\begin{array}{cc}
				\dfrac{\partial \varphi}{\partial K} & \dfrac{\partial \varphi}{\partial L}\\*[\baselineskip]
				\dfrac{\partial \psi}{\partial K}      & \dfrac{\partial \psi}{\partial L}
			\end{array}
		\right| \neq 0; \forall t.
	\]
So that Equation (\ref{EQ:Tt}) can be solved for $ K $ and $ L $ to receive 
     \begin{equation}\label{EQ:Tt1}
		\bar{K} = K(t) = K_{t},  \, \, \, \bar{L} = L(t) = L_{t}  \, \, \, \mathrm{and}  \, \, \, T = \{T_t = (K_t, L_t) \}
	\end{equation} 
\end{defn}

\begin{assum}\label{Ass122}{\bf (The Lie Group Properties)}\\
For any technical progress $T = \{T_{t} \}$  in (\ref{EQ:Tt}) or (\ref{EQ:Tt1}), we always assume that it possesses the Lie group properties, i.e. $T = \{T_{t} \}$ satisfies the following conditions:
\begin{enumerate}
	\item[(1)] $ T_{t} \circ T_{\bar{t}} = T_{t+\bar{t}}; \, \, t \in \R $;
	\item[(2)] $ T_{t}^{-1} = T_{-t}; \, \, t \in \R $;
	\item[(3)] $ T_{0} = Identity $.
\end{enumerate}
\end{assum}

\begin{defn}\label{D123}{\bf (The Lie Type of Technical Progress)}\\
	Every family of technical progress functions $T = \{T_{t} \}$ in (\ref{EQ:Tt}) or (\ref{EQ:Tt1}) which satisfies the conditions (1), (2), (3) in Assumption \ref{Ass122} is called {\it a technical progress of the Lie type} or {\it a Lie type technical progress}.
\end{defn}

Note that, any Lie type technical progress $T = {\{T_{t}\}}_{t \in \R} $ always forms an one-parameter continuous subgroup of transformations of an appropriate Lie group $G$. In particular, when $t$ belongs to an finite subset of the straight line $\R$, then $T = \{T_{t} \}$ forms a finite one-parameter continuous subgroup of $G$. Therefore, we can {\it apply Lie groups in general, one-parameter subgroup of transformations of Lie groups in particular to study the technical progress functions in Economy}. 

\begin{rem}\label{R124} It may be argued that Assumptions \ref{Ass122}  are too restrictive, for there may be many types of technical progress operating in an economy which do not satisfy these restrictions. {\it However, it should be noted that all of the known types discussed in the economic literature thus far do in fact satisfy the foregoing Assumption} \ref{Ass122} (see \cite{rS1981}, p.27).
\end{rem}

\subsection{Holotheticity of a Production Function under a Given Type of Technical Progress}

\begin{defn}\label{D131}({\cite[Definition 2]{rS1981}})\\
If the action of technical progress $T = {\{T_{t} = (K_{t}, L_{t}) \}}$ on a production function $Y = f (K, L)$ is represented by some family $\{F_t (Y)\}$ of transformations $ F_t (Y) $ which is strictly monotone with respect to the parameter $t$, then the production function $Y = f (K, L)$ is said to be ``holothetic" (complete-transformation type) under the given technical progress $T$.

In other words, a production function $Y = f (K, L)$  is said to be {\it holothetic under a given technical progress $T = {\{T_{t} = (K_{t}, L_{t}) \}}$} if there exists a family $\{F_t (Y)\}$ of transformations $ F_t (Y) $ which is strictly monotone with respect to the parameter $t$ such that
      \begin{align}\label{EQ:Yt}
                    Y_{t} &= f(K_{t}, L_{t}) \notag\\
                             &= F_{t}\Bigl(f(K, L)\Bigr) = F_{t}(Y)
      \end{align}
for all parameter $t$ and $Y_0 = Y$, i.e. $f(K_0, L_0) = f(K, L)$.
\end{defn}

 \begin{rem}\label{R132}({\cite[page 24]{rS1981}}) 
           The following facts are immediate from the holotheticity of a production function under under a given technical progress:
            \begin{enumerate}
            \item[(1)] The total impact of technical progress is completely transformed to a scale effect. Hence the isoquant map, before and after technical progress, is unchanged other than the relabeling of its isoquants.
            \item[(2)] The marginal rate of substitution between capital and labor is unaffected by technical progress.
            \item[(3)] The technical progress functions transform every production function into another function of the same family. Hence a family of production functions is invariant under the technical progress transformation.
            \end{enumerate}
          \end{rem}
      
\subsection{The Results of Saito (see~\cite{rS1980}, \cite{rS1981})}

            \begin{prop}[{\bf Existence of Holothetic Production Function}]            	
                If $T = \{T_t\}$ in Equation (\ref{EQ:Tt}) is a given technical progress of the Lie type (i.e. it satisfies the Lie group properties in Assumptions \ref{Ass122}), then there exists one and only one   production function $Y = f(K, L)$ which is holothetic under $T = \{T_t\}$, such that Equation (\ref{EQ:Yt}) holds. Hence there exists a general family $\{Y_t\}$ of production functions under which the total effect of technical progress $T = \{T_t\}$ is completely transformed to a returns to scale effect. \hfill $\square$
            \end{prop}
        
            \begin{prop}[{\bf Possibility of the Estimation of Technical Progress}]
                The effect of technical progress $T = \{T_t\}$ and the scale effect of $Y = f(K, L)$ are independently identifiable if and only if the production function $f$ is not holothetic under the given type of technical progress $T$. \hfill $\square$
            \end{prop}
        
            \begin{prop}[{\bf Existence of a Lie Type Technical Progress}]
                For every production function $Y = f(K, L)$, there exists at least one Lie type of technical progress $T = \{T_t\}$ such that $f$ is holothetic under $T$. \hfill $\square$
            \end{prop}

\section{The Main Result}

In this section, we  will firstly introduce an example of Lie type technical progress $T = \{T_t\}$ and one type of production functions which is holothetic under $T$. That is exactly the Cobb-Douglas function with exogenous technical progress. Next, we introduce the illustration result of G. Zaman and Z. Goschin in \cite{ZG2010} and give the main result of the paper which is an empirical research in Viet Nam.

\subsection{The Cobb-Douglas Production Function as a Holothetic Function under Technical Progress}

\begin{defn}\label{D211}({\bf The Cobb-Douglas Production Function})\\
     Let $a, \alpha, \beta$ be some positive real constants. 
	The Cobb-Douglas production function is given by     
		\begin{equation}\label{EQ21}
			Y = f(K, L) : = a.K^{\alpha}L^{\beta}.
		\end{equation}
\end{defn}

The following proposition state one of the most important properties of the Cobb-Douglas production function.
\begin{prop}\label{Prop212}{\bf (An Lie Type Technical Progress under \\
		Which the Cobb-Douglas Is Holothetic)}
	\begin{enumerate}
		\item[(1)] {\bf (see \cite[page 2, 3]{rS1981})}  Let $T = \{T_t \}$ be the family of pairs $T_t = (K_t, L_t)$ of the following functions
			\begin{equation}\label{EQ22}
				K_t = e^{\lambda t}K; \,\, L_t = e^{\lambda t}L; \, \, \, t \in \R;
			\end{equation}
		\item[] where $\lambda$ is a some positive real constant. Then $T = \{T_t\}$ is a technical progress of the Lie type.
		
		\item[(2)] The Cobb-Douglas production function defined by Formular (\ref{EQ21}) is holothetic under the Lie type technical progress $T = \{T_t\}$ in Formular (\ref{EQ22}) above. 
		\hfill $\square$	
	\end{enumerate}
\end{prop}

\begin{proof}
	\begin{itemize}
		\item[(a)] {\bf The Proof of (1)}: From Formular (\ref{EQ22}), we have
		\[ K_t  = \varphi (K, L, t) : = e^{\lambda t}K, \, \, L_t  = \psi (K, L, t) : = e^{\lambda t}L; \, \, \forall t.\]
		\item[] It is obviuos that $\varphi$ and $\psi$ are analytic real functions of the three variables $ K, \, L, \, t$. Moreover, $ \varphi $ and $ \psi $  are independent functions with respect to $ K $ and $ L $ alone because
		\begin{equation*}
				\left|
		\def\arraystretch{2.2}
		\begin{array}{cc}
		\dfrac{\partial \varphi}{\partial K} & \dfrac{\partial \varphi}{\partial L}\\*[\baselineskip]
		\dfrac{\partial \psi}{\partial K}      & \dfrac{\partial \psi}{\partial L}
		\end{array}
		\right| 
			= \begin{vmatrix}
			e^{\lambda t} & 0\\
			0 & e^{\lambda t}
			\end{vmatrix}
			= e^{2 \lambda t} > 0; \forall t. \notag
		\end{equation*}
		\item[] Therefore, $T = \{T_t = (K_t, L_t)\}$ is a technical progress.
		\item[] Besides, it is clear that
		\[e^{\lambda (t + \bar{t})} = e^{\lambda t} e^{\lambda \bar{t}}, \,  {\bigl(e^{\lambda t}\bigr)}^{-1} = e^{\lambda (-t)}, \, e^{\lambda .0} = 1; \, \, \forall t, \bar{t}.\]
		Therefore $T = \{T_t = (K_t, L_t)\}$ satisfy the conditions (1), (2), (3) in Assumption \ref{Ass122}. So that $T = \{T_t = (K_t, L_t)\}$ is a Lie type technical progress.
		\item[(b)] {\bf The Proof of (2)}:	It follows from the homogeneity of the Cobb-Douglas function that
		\[Y_t = f(K_t, L_t) = f(e^{\lambda t} K, e^{\lambda t} L) = e^{\lambda t} f(K, L) =  F_t (Y); \forall t\]
		\item[] where $F_t (Y) : = e^{\lambda t}.Y$ forall $t$.
		\item[] It is obvious that $F_t (Y) : = e^{\lambda t}.Y$ is strictly monotone with respect to the parameter $t$ and $F_0 (Y) = Y$ because the exponent function $e^{\lambda t}$ has the same property.
		So that the Cobb-Douglas production function is holothetic under the given Lie type technical progress $T = \{T_t\}$.
	\end{itemize}		
The proof is complete.
\end{proof}
\subsection{The Illustration in Romania}
   As mentioned in Introduction, in 2010, G. Zaman and Z. Goschin \cite{ZG2010} developed several models to estimate the aggregated production function GDP of Romania for the 1990--2007 period, and to assess the impact of technical progress on the  economy growth of Romania. Basically, models can be divided into two types based on points of view: technical progress is endogenous or exogenous. 
   
   When treating technical progress as exogenous, the authors obtained the following results: 
   
    	   \begin{prop}[{\cite[Model 1]{ZG2010}}]\label{P221}
	For the 1990--2007 period, the GDP function of Romania is given by 
			\[
			\mathrm{GDP}(t) =0.021.e^{0.0105t}K^{0.3564}L^{0.7783}; \, \, t \geq 0. 
			\] \hfill $\square$
    		\end{prop} 

	   \begin{rem}\label{R222}
	Note that the production function model and the technical progress mentioned in Proposition \ref{P221} are the same as the production function and the technical progress shown in Formulars (\ref{EQ21}) and (\ref{EQ22}).        
    \end{rem}

    \subsection{Estimating the GDP Function of Viet Nam for the 1995-2018 Period} 
    As a problem for empirical research in Viet Nam, we also want to estimate the GDP function of Viet Nam for the 1995--2018 period by applying Saito's results and using the similar method of G. Zaman and Z. Goschin in \cite{ZG2010}. 
    
    However, stemming from the fact that, the technical progress in economy of Viet Nam mainly come from external transfers for many years, while the domestic investment in research and development (R\&D) is very limit. Namely, according to World Bank data, the R\&D investment of Viet Nam in 2011 and earlier did not exceed 0.2\% of GDP and only increased approximately to 0.53\% of GDP in 2018. Therefore, models in which technical progress is endogenous are inappropriate. So we just need to investigate the model in which technical progress is exogenous.
    
\begin{assum}\label{Ass231}

    In view of Proposition \ref{Prop212}, we will use the model of the Cobb-Douglas production function for GDP of Viet Nam. Based on Formulars (\ref{EQ21}) and (\ref{EQ22}), the aggregated production function GDP of Viet Nam is accepted as the following formula:
	\begin{equation}\label{EQ23}
	\mathrm{GDP}(t) = a. e^{ \gamma t} K^{\alpha} L^{\beta} 
	\end{equation}
	where $a, \,  \alpha, \, \beta, \, \gamma$ are certain positive contants; $t$ is time parameter. 
\end{assum}  
Thus, our task is to use ``skillfully" the statistical data of Viet Nam to analyze and estimate the parameters $a, \, \alpha, \, \beta, \, \gamma$ in Formula (\ref{EQ23}).

\begin{dataso}\label{Dataso232}
We can find data on rate of GDP increase, labor and investment at Annual Statistical Report issued by the Viet Nam National General Statistics Office.	
\end{dataso}
\noindent However, by the way, there are no statistics about the capital in Viet Nam. Therefore, the amount of the capital for each year can be calculated by:
	 \begin{equation}\label{EQ24}
        K(t) = K(t-1) + I (t) - \sigma ( I(t)/2+K(t-1) ). 
       \end{equation}
Here, $K(t)$ is the amount of the capital for year $t$, $\sigma$ is fixed asset depreciation rate, $I(t)$ is the investment for the year $t$ with $t$ from 1995 to 2018. Based on data provided by Viet Nam National General Statistics Office, we calculate the capital $K$ by the formula (\ref{EQ24}).
	 
	About the labor $L$, of course we only count the labor force from 16 years old and get paid on their job. The Stata software is used to filter data and eliminate the trend of the data before analysing data. 

From Formular (\ref{EQ23}), taking the logarithm we get
	 \begin{equation}\label{EQ25}
		\ln(\mathrm{GDP}(t)) = \gamma t + \alpha \ln(K) + \beta \ln(L) + \ln a 
	\end{equation}

Combining two data groups: before 2000 and after 2000, and then analyzing them by Stata software, we get the values of estimating the parameters of the GDP function in Table 1 below. 
	\begin{figure} [!htp]
		\centering 
		\includegraphics [scale=0.84] {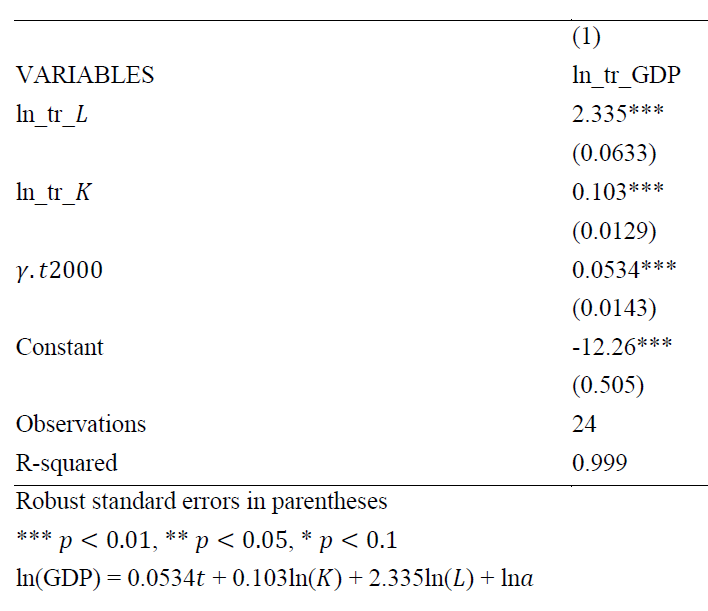}
		\centerline{Table 1: The results of estimating the parameters.}
	\end{figure}

Now, the next theorem, which is the main result of the paper, is followed immediately.
    	   \begin{thm}\label{T233}
    		For the 1995-2018 period, the aggregated production function GDP of Viet Nam is given by \[\mathrm{GDP}(t) \approx 0.000005.e^{0.053t}K^{0.103}L^{2.335}; \, \, t \geq 0.\]
    	  \end{thm} 

	\begin{rem} From the main result, we have some important remarks as follows:
	\begin{enumerate} 
		\item[(1)] Contribution rate of technology, capital $K$ and labor force $L$ on GDP growth:
		\begin{itemize}
		\item[(a)] The contribution rate of technology: $\gamma / ( \gamma + \alpha + \beta )$ = 2,1\%.
		
		\item[(b)] The contribution rate of capital $K$: $\alpha / ( \gamma + \alpha + \beta )$ = 4,2\%.
		
		\item[(c)] The contribution rate of  labor force $L$: $\beta / ( \gamma + \alpha + \beta )$ = 93,7\%.
		\end{itemize}
	
		\item[(2)] The labor force $L$ has the strongest impact on GDP growth. This is entirely reasonable in the context of the economy and society of Viet Nam. In fact, for many years, the labor costs in Viet Nam are lower than the average labor costs in the world, there is the so-called labor-intensive manufacturing in production. The cheap labor in Viet Nam for many years was one of the competitive advantages and an important engine of the economic growth.
		
		\item[(3)] The contribution of the capital $K$ is quite limited compared to the labor force $L$. It partly comes from the fact that a lot of spending in the population is not calculated precisely, since it is difficult to control cash consumption in the market (not through the banking system). Therefore, the calculation of $K$ is not accurate compared to reality.
		
		\item[(4)] The contribution of technical progress is the lowest. That is quite reasonable because the economy of Viet Nam, for many years ago, was mainly based on the processing, assembling and exporting raw materials. This also warns the managers in formulating necessary policies to intensify the contribution of science and technology towards the knowledge economy in the context of the 4.0 technological revolution in the world.
		
		\item[(5)] The constant $a$ (total factor productivity) is too small. This is predictable because it represents the productivity of the Vietnamese economy at the time of departure (1995), when there was almost no influence of technology and the size of the economy is too small. 
	\end{enumerate}
	\end{rem}

\section*{Acknowledgements} The authors would like to thank Viet Nam National University - Ho Chi Minh City, University of Economics and Law, VNU-HCMC and Hoa Sen University for financial support.

    
\end{document}